\documentclass[11pt]{article}
\usepackage[margin=1in]{geometry}
\usepackage{amsmath,amssymb,amsthm}
\usepackage{booktabs,tabularx,multirow}
\usepackage{graphicx}
\usepackage{xcolor}
\usepackage{tikz}
\usetikzlibrary{arrows.meta,positioning,fit,backgrounds,calc}
\usepackage[most]{tcolorbox}
\usepackage{algorithm}
\usepackage{algpseudocode}
\usepackage{enumitem}
\usepackage{hyperref}
\hypersetup{colorlinks=true,linkcolor=blue!50!black,citecolor=blue!50!black,urlcolor=blue!50!black}
\setlist{itemsep=2pt,topsep=4pt}
\newtheorem{theorem}{Theorem}
\newtheorem{definition}{Definition}

\newcommand{\pgs}{\textsc{PGS}}
\newtcolorbox{keybox}[1]{colback=blue!3,colframe=blue!45!black,boxrule=0.6pt,arc=2pt,left=6pt,right=6pt,top=4pt,bottom=4pt,title={#1},fonttitle=\bfseries\small,coltitle=white}

\title{\textbf{Proof-Gated Signing:}\\[2pt] \Large Solver-Checked Transaction Guards that Hold Under State Drift\\ for Onchain AI Agents}
\author{Bravish Ghosh\\ \small Independent Researcher}
\date{\small Preprint --- September 2026}

\begin{document}
\maketitle

\begin{abstract}
AI agents that control wallets read attacker-reachable content (emails, documentation, tool outputs), so they can be steered into proposing harmful transactions. The usual last line of defense is a pre-signing check: a static allowlist, an LLM reviewer, or a transaction simulation. We point out a gap that all three share. A pre-signing check is evidence about the chain state \emph{at check time}, but the transaction executes in a later state that an adversary can shape through front-running, contract upgrades or token-parameter changes. We call the difference \emph{state drift}.

We present \emph{Proof-Gated Signing} (\pgs{}). \pgs{} simulates a proposed transaction, extracts its effects on the wallet, and uses an SMT solver to check a declarative value-and-permission policy for every price in an oracle-uncertainty band. It then compiles on-chain post-conditions and proves that \emph{every} execution satisfying them also satisfies the policy. The post-conditions cover wallet balance bounds, payee receipts, allowance caps and ownership. The agent's smart-contract wallet enforces them atomically, so the guarantee is about the transaction that actually executes, and it holds under arbitrary drift.

We evaluate \pgs{} on an open testbed of 260 scripted scenarios: 14 attack families, including five state-drift families and two that adapt to the defense, and 12 benign families. Harm is measured from attacker balances, never from the policy. \emph{On this suite}, \pgs{} prevented 131 of 140 harmful scenarios (93.6\%, 95\% CI 88--97\%) and passed 117 of 120 benign ones (97.5\%). Simulation-only checking prevented 57.9\% and a static allowlist 71.4\%. None of the 50 drift scenarios produced attacker gain under \pgs{}. In 40 of them the transaction reverted on-chain, and in the other 10 an attacker who read the post-conditions from the mempool found no profitable attack within them. The nine scenarios \pgs{} did not prevent come from one adaptive family that drains value within per-transaction policy; the per-session budget bounds it (at most \$4.76k against a \$5k budget). The median overhead is 41k gas and about 0.1--0.2\,s per check. An independent audit of our first draft found a payment-crediting flaw, which we fixed and quantify by ablation. Code, contracts and all raw logs are released, and a rerun reproduced all 1{,}300 runs exactly.
\end{abstract}

\section{Introduction}
Autonomous agents are increasingly trusted with money. Treasury bots pay invoices, trading assistants rebalance portfolios, and ``agentic wallets'' act on natural-language instructions. Every one of these agents reads content an attacker can influence: an email, a token's metadata, a governance post, the output of a quoting tool. Indirect prompt injection~\cite{greshake2023not} and the benchmarks that followed~\cite{zhan2024injecagent,debenedetti2024agentdojo} show that such content can redirect agents. For an agent that signs transactions, the question is therefore not \emph{whether} it can be misled, but \emph{what is guaranteed about the transactions it signs when it is}.

The last line of defense is usually a check performed just before signing. It can take three forms. A static policy (an allowlist of contracts, functions and recipients). A second model that reviews the transaction (an LLM judge). Or a simulation that previews what the transaction will do (the ``you will send X and receive Y'' preview in modern wallets). These checks are useful, but they share a structural blind spot, and recent writing on formal methods for onchain systems frames it well: a proof is only as strong as its property, its assumptions and the version of the system it applies to~\cite{buterin2026fv,sherlock2026fv}. A pre-signing check is a claim about the chain \emph{as it was when checked}. The transaction executes later, in a state the adversary can influence. A searcher can sandwich the trade~\cite{daian2020flash,qin2022quantifying}. The admin of an upgradeable contract can swap its implementation. The owner of a fee-on-transfer token can raise the fee. In each case the check was correct and the executed transaction is still harmful. We call this \emph{state drift}.

\paragraph{Idea.} \pgs{} treats the pre-signing check as a proof with an explicit scope, and it closes the drift gap by \emph{compiling the proof's assumptions into on-chain post-conditions}. The guard simulates the transaction, extracts an effect vector and proves a policy over it. It then derives bounds (``the wallet's WETH balance must increase by at least $c$''; ``Alice must receive at least $\pi$ USDC'') and proves that \emph{any} outcome within those bounds satisfies the policy. The wallet executes the transaction together with these bounds in one atomic call and reverts if any bound fails. If the chain drifts, the transaction either still satisfies the policy or does not happen.

\begin{keybox}{What is new, stated precisely}
The individual ingredients are established: transaction simulation, SMT solving, and post-execution checks such as a swap's \texttt{minOut}. Our contribution is their combination into one guard with a stated guarantee for agent-signed transactions:
\begin{enumerate}[leftmargin=14pt]
\item \textbf{Drift as the threat.} We identify state drift as the failure mode shared by simulation-, allowlist- and LLM-based pre-signing guards. We show it empirically, including the finding that a clean simulation makes an LLM reviewer \emph{more} likely to approve a drift attack.
\item \textbf{Proof-carrying post-conditions.} Post-conditions are not heuristics. They are derived so that a solver-checked implication ``post-conditions $\Rightarrow$ policy, for all prices in the band'' holds. That turns a claim about a simulated state into a claim about the executed transaction (Theorem~\ref{thm:sound}). Existing post-execution checks (e.g., \texttt{minOut}) protect one asset of one swap, and none is derived from a whole-wallet policy.
\item \textbf{An honest scope.} We separate what is \emph{prevented}, what is \emph{bounded} and what is \emph{uncovered}, and we evaluate on attacks designed to exploit exactly those boundaries: mempool-aware drift and in-policy drains.
\item \textbf{An open, reproducible testbed} with ground truth independent of any defense, per-family results and all raw logs.
\end{enumerate}
\end{keybox}

\paragraph{Results in one paragraph.} On our 260-scenario suite, \pgs{} prevented 93.6\% of harmful scenarios and passed 97.5\% of benign ones. Every drift scenario (50/50) ended with no attacker gain. Simulation-only checking missed all 50, and the static allowlist missed 40. The only attacks \pgs{} failed to prevent were in-policy drains, and there the loss stayed under the session budget. We stress that these rates describe \emph{this} suite, which we designed; Section~\ref{sec:limits} discusses what they do and do not imply.

\paragraph{Paper outline.} Section~\ref{sec:model} defines the threat model, the policy and the soundness theorem, and puts the guarantee boundary in one table. Section~\ref{sec:design} describes the design. Section~\ref{sec:eval} presents the evaluation. Section~\ref{sec:related} positions \pgs{} against slippage protection, simulation, allowlists, runtime and formal verification, account abstraction and intent-based systems. Section~\ref{sec:limits} covers limitations.

\section{Model and Guarantee}
\label{sec:model}

\subsection{System and threat model}
An \emph{agent} proposes call batches $b=(c_1,\dots,c_k)$ to be executed by a smart-contract \emph{wallet} $W$. A \emph{guard} sits between the agent and the signing key. It sees each proposed batch and either refuses it or signs a call to $W.\texttt{executeChecked}(b,\Phi)$, where $\Phi$ is a set of post-conditions.

\paragraph{Adversary.} The adversary can (i) inject arbitrary content into anything the agent reads, so the agent may propose \emph{any} batch; and (ii) act on-chain between the guard's check and the batch's inclusion. That includes front-running and back-running, upgrading contracts it administers, and changing parameters of tokens it owns. It can read pending transactions, including their post-conditions.

\paragraph{Trusted components.} (a) The guard's code and its custody of the signing key: the agent cannot sign on its own. (b) The wallet contract's implementation of \texttt{executeChecked}. (c) The \emph{tracked} token contracts report balances truthfully. (d) The oracle: the true price of each tracked asset lies within the stated band. (e) The policy itself reflects what the owner wants.

\subsection{Effects and policy}
Let $A$ be the tracked (priced) assets, including native ETH and vault shares. For a batch execution, let $\Delta_a$ be the change in the wallet's balance of $a$, and let $\pi_a$ be the amount \emph{credited} as a payment to allowlisted payees (defined in Section~\ref{sec:design}). Prices are $p_a\in[\underline p_a,\overline p_a]=[p^0_a(1-\beta),p^0_a(1+\beta)]$ for nominal price $p^0_a$ and band $\beta$.

\begin{definition}[Value policy]
Unexplained loss and outflow value are
\[
L(\Delta,p)=-\sum_{a\in A}\Delta_a p_a-\sum_{a\in A}\pi_a p_a,\qquad
V(\Delta,p)=\sum_{a\in A}\big(\max(-\Delta_a,0)-\pi_a\big)p_a .
\]
A batch satisfies the value policy if $L\le s\,V+f$ (relative loss bound $s$ with floor $f$) and $W_{\text{spent}}+L\le W_{\max}$ (session budget).
\end{definition}

\begin{definition}[Structural policy]
(S1) Value sent to externally owned accounts (EOAs) goes only to allowlisted payees, within per-payee caps. (S2) No untrusted spender retains a positive allowance after the batch. (S3) Contracts owned by the wallet keep their owner. (S4) No flow of unpriced tokens to or from the wallet.
\end{definition}

The relative bound lets the agent pay ordinary trading costs (fees, price impact) while refusing transactions that give away value. The budget caps what a sequence of individually acceptable transactions can lose.

\subsection{Soundness}
Let $c\in\mathbb Z^{A}$ be lower bounds on balance changes, $R=\{(a,q,\pi_{a,q})\}$ payee-receipt bounds, $C=\{(a,\sigma,m)\}$ allowance caps, and $O$ the set of owned contracts.

\begin{theorem}[Soundness of assert mode]\label{thm:sound}
Suppose the solver has established
\[
\forall p\in[\underline p,\overline p]\ \ \forall \Delta\ge c:\quad L(\Delta,p)\le s\,V(c,p)+f\ \wedge\ W_{\text{spent}}+L(\Delta,p)\le W_{\max},
\]
and $W.\texttt{executeChecked}(b,\Phi)$ with $\Phi=(c,R,C,O)$ completes without reverting. Then, for every price vector in the band, the realized balance changes satisfy the value policy. Every credited payee $q$ received at least $\pi_{a,q}$ of asset $a$. Every capped allowance is at most its cap. Every owned contract keeps its owner. This holds for any sequence of state changes between simulation and inclusion.
\end{theorem}
\begin{proof}[Proof sketch]
\texttt{executeChecked} records the wallet's and payees' balances, runs the batch, and reverts unless $\Delta_a\ge c_a$ for all $a$, every receipt bound, every allowance cap and every owner check holds. All of this happens in one transaction, so no state change can interleave. A non-reverting execution therefore yields $\Delta\ge c$ and receipts $\ge\pi$. The solver's statement applies at the realized $\Delta$. Moreover $L-sV$ is non-increasing in each $\Delta_a$ when $s<1$, so bounding $V$ at $c$ is conservative. Crediting exactly $\pi$ when payees receive at least $\pi$ is conservative for $L$ as well.
\end{proof}

\begin{table}[t]
\centering\small
\caption{\textbf{The guarantee boundary.} What Theorem~\ref{thm:sound} covers, and what it does not.}
\label{tab:boundary}
\begin{tabularx}{\textwidth}{>{\raggedright\arraybackslash}p{0.30\textwidth}X}
\toprule
\textbf{Guaranteed (for executed transactions)} & \textbf{Why} \\
\midrule
Value policy on tracked assets, for all prices in the band & Solver-proved implication plus on-chain balance bounds \\
Credited payees receive at least the credited amount & On-chain receipt checks \\
Checked allowances do not exceed their simulated values & On-chain allowance caps \\
Owned contracts keep their owner & On-chain owner checks \\
All of the above under arbitrary state drift & Checks are atomic with the batch \\
\midrule
\textbf{Bounded, not prevented} & \\
\midrule
Loss within per-transaction policy (fees, in-policy drains) & Only the session budget $W_{\max}$ limits it, per session \\
Drift within the post-condition tolerance $\tau$ and the band & Allowed by design; worst case is at most $\tau$ of the inflow plus band slack \\
\midrule
\textbf{Not covered} & \\
\midrule
Untracked assets (NFTs, positions without a balance view) & Not in the effect vector; unpriced token flows are refused instead \\
Off-chain signatures (e.g., permits) & Never reach the transaction path \\
S1 for internal ETH transfers; S2 for allowances never observed & Checked on the simulation only, and only for observed pairs \\
A wrong policy, a wrong oracle band, a dishonest tracked token & Trusted inputs (Section~\ref{sec:model}) \\
An agent holding the key and calling the unguarded path & The deployment must give key custody to the guard \\
\bottomrule
\end{tabularx}
\end{table}

\paragraph{Reading the theorem.} Table~\ref{tab:boundary} is the most important table in this paper. The theorem is narrow on purpose. It says nothing about untracked assets, and it cannot tell an adaptive attacker who stays within policy from an unlucky trade. It relies on key custody. What it does say is strong: whatever happens on-chain between check and inclusion, a transaction that executes respects the proven bounds.

\section{Design}
\label{sec:design}

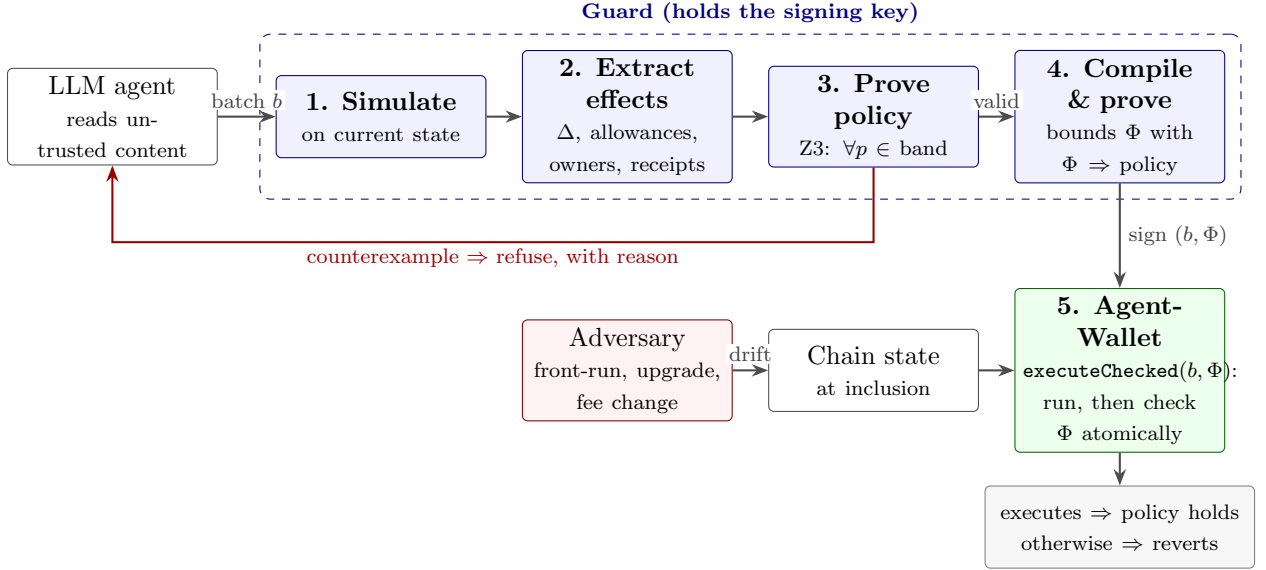
\begin{figure}[t]
\centering
\resizebox{\textwidth}{!}{%
\begin{tikzpicture}[
  box/.style={draw=black!70, rounded corners=2pt, align=center, font=\small, minimum height=11mm, text width=26mm, inner sep=3pt, fill=white},
  gbox/.style={box, fill=blue!6, draw=blue!50!black},
  cbox/.style={box, fill=green!7, draw=green!40!black},
  abox/.style={box, fill=red!5, draw=red!50!black},
  arr/.style={-{Stealth[length=2.2mm]}, thick, black!70},
  lab/.style={font=\scriptsize, black!75, align=center, fill=white, inner sep=1pt}
]
\node[box]  (agent)   at (0,0)     {LLM agent\\\scriptsize reads untrusted content};
\node[gbox] (sim)     at (3.6,0)   {\textbf{1. Simulate}\\\scriptsize on current state};
\node[gbox] (eff)     at (6.9,0)   {\textbf{2. Extract effects}\\\scriptsize $\Delta$, allowances, owners, receipts};
\node[gbox] (prove)   at (10.2,0)  {\textbf{3. Prove policy}\\\scriptsize Z3: $\forall p\in$ band};
\node[gbox] (compile) at (13.5,0)  {\textbf{4. Compile \& prove}\\\scriptsize bounds $\Phi$ with $\Phi\Rightarrow$ policy};
\node[abox] (adv)     at (6.9,-3.4) {Adversary\\\scriptsize front-run, upgrade, fee change};
\node[box]  (chain)   at (10.2,-3.4) {Chain state\\\scriptsize at inclusion};
\node[cbox] (wallet)  at (13.5,-3.4) {\textbf{5. AgentWallet}\\\scriptsize \texttt{executeChecked}$(b,\Phi)$: run, then check $\Phi$ atomically};
\node[box, draw=black!50, fill=black!3, text width=34mm] (out) at (13.5,-5.5) {\scriptsize executes $\Rightarrow$ policy holds\\ \scriptsize otherwise $\Rightarrow$ reverts};
\begin{scope}[on background layer]
\node[fit=(sim)(eff)(prove)(compile), draw=blue!50!black, dashed, rounded corners=4pt, inner sep=6pt,
      label={[font=\scriptsize\bfseries,blue!50!black]above:Guard (holds the signing key)}] (guard) {};
\end{scope}
\draw[arr] (agent) -- node[lab,above=2pt]{batch $b$} (sim);
\draw[arr] (sim) -- (eff);
\draw[arr] (eff) -- (prove);
\draw[arr] (prove) -- node[lab,above=2pt]{valid} (compile);
\draw[arr] (compile) -- node[lab,right=2pt]{sign $(b,\Phi)$} (wallet);
\draw[arr] (adv) -- node[lab,above=2pt]{drift} (chain);
\draw[arr] (chain) -- (wallet);
\draw[arr] (wallet) -- (out);
\draw[arr, red!60!black] (prove.south) -- ++(0,-1.0) -| node[lab, pos=0.25, below=1pt, text=red!60!black]{counterexample $\Rightarrow$ refuse, with reason} (agent.south);
\end{tikzpicture}}
\caption{\textbf{Architecture of Proof-Gated Signing.} The guard checks a policy on the simulated effects (steps 1--3). It then compiles post-conditions $\Phi$ whose satisfaction provably implies the policy (step 4). The wallet enforces $\Phi$ atomically with the batch (step 5), so drift between check and inclusion either leaves the policy satisfied or reverts the transaction.}
\label{fig:arch}
\end{figure}

Figure~\ref{fig:arch} shows the pipeline, and Algorithm~\ref{alg:pgs} gives it in pseudocode.

\begin{algorithm}[t]
\caption{Proof-Gated Signing for one proposed batch $b$}
\label{alg:pgs}
\small
\begin{algorithmic}[1]
\State $\sigma \gets \textsc{Snapshot}()$; execute $b$ from the wallet on the current state; read effects $E$; \textsc{Revert}$(\sigma)$
\If{$b$ reverted in simulation \textbf{or} $E$ violates a structural rule S1--S4} \Return \textsc{Refuse}(reason) \EndIf
\If{$\textsc{Solve}\big(\exists p\in\text{band}: \neg\text{Policy}(\Delta^{\text{sim}},p)\big)=\textsc{sat}$} \Return \textsc{Refuse}(counterexample) \EndIf
\State $c_a \gets \Delta^{\text{sim}}_a$ if $\Delta^{\text{sim}}_a\le 0$, else $\lfloor(1-\tau)\Delta^{\text{sim}}_a\rfloor$ \Comment{outflows exact, inflows with tolerance}
\If{$\textsc{Solve}\big(\exists p,\ \exists\Delta\ge c: \neg\text{Policy}(\Delta,p)\big)=\textsc{sat}$} $c\gets\Delta^{\text{sim}}$ \Comment{retry with $\tau=0$}
  \If{still \textsc{sat}} \Return \textsc{Refuse}(unprovable) \EndIf
\EndIf
\State $\Phi\gets(c,\ \text{receipts } \pi,\ \text{allowance caps},\ \text{owner checks})$
\State \Return \textsc{Sign}$\big(W.\texttt{executeChecked}(b,\Phi)\big)$
\end{algorithmic}
\end{algorithm}

\subsection{Effect extraction}
The guard takes a snapshot of a local fork, executes the batch from the agent key through the wallet's unguarded \texttt{execute} entry point, reads the effects, and reverts the snapshot. The effect vector contains:
\begin{itemize}
\item \textbf{Balance changes} of all tracked assets: the ERC-20 tokens in the price table, native ETH, and shares of owned vaults (priced at the vault's share price at session start).
\item \textbf{Residual allowances} for every (token, spender) pair seen in the simulation's \texttt{Approval} logs, or granted earlier in the same guard session.
\item \textbf{Owners} of contracts the wallet owns.
\item \textbf{Value sent to EOAs} (from \texttt{Transfer} logs and top-level ETH sends), and \textbf{unpriced tokens} moving to or from the wallet.
\item \textbf{Payment credits.} A top-level, direct token transfer or ETH send to an allowlisted payee is a \emph{candidate} payment. The credit is the smaller of the intended amount and the amount the payee actually received in simulation. Fees charged by the token are therefore counted as loss, not as payment.
\end{itemize}
The last rule was introduced after an independent audit of our first draft. That draft credited the intended (calldata) amount, which let a fee-on-transfer token route part of a ``payment'' to the token's owner without the loss appearing. Family H14 and an ablation in Section~\ref{sec:eval} quantify the flaw.

\subsection{Proof obligations}
For each tracked asset the solver gets a real price variable $p_a$ constrained to the band and a real variable $\Delta_a$. In \emph{check} mode $\Delta_a$ is fixed to the simulated value. In \emph{compile} mode it is constrained by $\Delta_a\ge c_a$. The solver checks satisfiability of the negated policy. \textsc{unsat} means the policy is valid over the whole region. \textsc{sat} returns a model: a concrete price vector and outcome that violates the policy, which also names the violated clause (relative bound or budget). The guard reports that model to the agent as the refusal reason.

\paragraph{Why an SMT solver?} For the policy used here the obligation is bilinear, and its worst case lies at a corner of the price-by-outcome box. A hand-written corner enumeration would decide it too, so the solver is not \emph{necessary} for this policy, and we do not claim otherwise. We use it for three reasons. (1) \emph{One verified decision procedure} for both obligations (check and compile), instead of hand-derived case analyses that must be re-derived and re-audited whenever the policy changes. (2) \emph{Counterexamples}: a \textsc{sat} answer is a concrete witness (``at WETH $=2{,}487.5$ and USDC $=1.005$, the loss is \$3{,}458'') that the agent or a human can act on. (3) \emph{Extensibility}. Realistic policies add cross-asset constraints (``net stablecoin exposure $\ge x$''), correlated price models (a joint band instead of independent intervals), piecewise fee schedules and per-counterparty budgets. For these the corner argument no longer applies, but the same encoding still works. In our runs the solver was never the bottleneck (Section~\ref{sec:overhead}).

\subsection{Post-condition compilation}
Outflows may not exceed their simulated amounts. Inflows must reach $(1-\tau)$ of their simulated amounts. Untouched tracked assets may not decrease. The guard asks the solver to prove the policy over the whole region $\Delta\ge c$. If that fails because the tolerance leaves too much room, it retries with $\tau=0$ (the executed outcome must be at least as good as the simulated one), and otherwise refuses. Payee credits become receipt checks, residual allowances become caps, and owned contracts become owner checks.

\subsection{Implementation and parameters}
The wallet (\texttt{AgentWallet}) exposes \texttt{executeChecked(calls, balanceChecks, allowanceChecks, ownerChecks, receiptChecks)}. It records the relevant balances, runs the calls, and reverts with a reason (\texttt{guard:balance}, \texttt{guard:receipt}, \ldots) if any check fails. The guard is about 300 lines of Python (web3.py~8.0, the Z3 solver~\cite{demoura2008z3} via z3-solver~5.1) and runs against a Hardhat~2 node (Cancun). Contracts are Solidity~0.8.26. Parameters in all experiments: band $\beta=\pm0.5\%$, $s=3\%$, $f=\$25$, $W_{\max}=\$5{,}000$ per session, $\tau=1\%$, payee caps \$20k (Alice) and \$10k (Bob), and trusted spenders = three verified pools and the treasury vault.

\section{Evaluation}
\label{sec:eval}
We ask four questions. \textbf{Q1:} How often does each defense prevent harm on attacks it was not tuned to, and at what cost to benign activity? \textbf{Q2:} What happens under state drift, including drift by an adversary who knows the defense? \textbf{Q3:} What remains uncovered, and is it bounded? \textbf{Q4:} What does the guard cost?

\subsection{Testbed}
\paragraph{World.} The world has tokens USDC, WETH and DAI, and TAX, a fee-on-transfer token owned by the attacker with a 1\% fee paid to an attacker-controlled contract. There are five honest constant-product pools: USDC/WETH, DAI/USDC, WETH/DAI, TAX/USDC, and a second USDC/WETH pool that is \emph{not} allowlisted. A treasury vault is owned by the wallet. Adversarial contracts: a fake pool with honest quotes that pays out 1\%, an upgradeable ``low-fee'' pool whose admin is the attacker, an airdrop ``claim'' that sweeps approvals, a fake bridge, and an adaptive ``skim'' pool that pays 98.5\% of a fair quote. The wallet holds about \$1.85M across assets.

\paragraph{Scenarios.} A scenario consists of a legitimate user task, the batches a possibly misled agent proposes, and adversary hooks. Hooks run at three moments: before the agent acts (setup), \emph{after the guard's check and before inclusion} (drift), and after the episode (for example, sweeping a granted approval). Each family is instantiated 10 times with seeded random amounts, giving 260 scenarios. Table~\ref{tab:families} lists the families.

\begin{table}[t]
\centering\small
\caption{Scenario families (10 variants each). ``Drift'' families act between the guard's check and inclusion; ``adaptive'' families are designed with knowledge of \pgs{}.}
\label{tab:families}
\begin{tabularx}{\textwidth}{llX}
\toprule
ID & Kind & Description \\
\midrule
H1 & intent & Invoice paid to the attacker instead of the payee \\
H2 & intent & Payment to a look-alike (address-poisoning) address \\
H3 & intent & ``Airdrop claim'': unlimited approval + claim that sweeps it \\
H4 & intent & Latent unlimited approval, swept after the episode \\
H5 & intent & Swap on a fake pool with honest-looking quotes \\
H6 & drift & Sandwich of a swap without slippage protection (profit-maximizing searcher) \\
H7 & drift & Upgradeable pool's implementation swapped to a draining one before inclusion \\
H8 & drift & TAX fee raised to 30--90\% before a TAX sale is included \\
H9 & intent & Vault ownership transferred to the attacker, vault swept afterwards \\
H10 & intent & Legitimate swap batch with a hidden small transfer to the attacker \\
H11 & intent & ETH ``bridge'' deposit into a contract that forwards to the attacker \\
H12 & adaptive & Slow drain: 20--30 swaps on the skim pool, each within per-transaction policy \\
H13 & drift, adaptive & Mempool-aware sandwich: the searcher reads the pending post-conditions and front-runs as far as they allow \\
H14 & drift & Payment in TAX: fee set to 0 before the check, raised to 30--90\% before inclusion \\
\midrule
B1--B3 & benign & Swaps on verified pools (USDC$\to$WETH, WETH$\to$USDC, DAI$\to$USDC) \\
B4, B5 & benign & Payments to allowlisted payees (USDC, ETH) \\
B6, B7 & benign & Vault deposit and withdrawal \\
B8 & benign & Swap on an honest but \emph{non-allowlisted} pool \\
B9 & benign & Unlimited approval to a trusted pool, then swap \\
B10 & benign & Multi-hop swap USDC$\to$WETH$\to$DAI in one batch \\
B11 & benign & Swap with \texttt{minOut}$=0$ and no MEV present (sloppy but harmless) \\
B12 & benign & High-impact swap (1--3\% cost) on the shallower pool \\
\bottomrule
\end{tabularx}
\end{table}

\paragraph{Ground truth.} A scenario is \emph{harmful} if, under \emph{no defense}, value moves to attacker-controlled or burned accounts (more than \$1 at nominal prices) or the vault's owner changes. The attacker set is the attacker EOA, the searcher, the look-alike address, the TAX fee sink, the attacker's fake and skim pools, and the burn address. The definition never refers to any policy. The declared labels matched ground truth in all 260 scenarios. Each scenario runs against a fresh guard instance, so budgets do not carry across scenarios. Multi-transaction sessions are studied separately (Section~\ref{sec:liveness}).

\paragraph{Defenses compared.}
\begin{itemize}
\item \textbf{None.}
\item \textbf{Static allowlist:} allowed target contracts, function selectors, approval spenders, transfer recipients and ETH recipients. It resembles the policy engines in many wallet products.
\item \textbf{\pgs{}-sim:} simulation, structural checks and the solver-checked value policy, but no post-conditions. It is the strongest ``simulate then sign'' guard we could build.
\item \textbf{\pgs{}:} the full system.
\item \textbf{\pgs{} w/o receipts:} an ablation without payee-receipt checks.
\item \textbf{LLM judge} and \textbf{LLM judge + simulation}: an \emph{intent-level} baseline. A frontier LLM (Anthropic Claude, model identifier \texttt{claude-opus-5-5}) sees the task and the decoded calls, with known entities labeled as a wallet would label them. It answers ALLOW or DENY, and in the second variant it also sees the simulated effects. We include it because LLM review is a common proposal for agent oversight, not because \pgs{} competes with it at the same layer. The two are complementary (Section~\ref{sec:judge}).
\end{itemize}
Metrics are the \emph{attack prevention rate} (APR: harmful scenarios with no harm) and the \emph{benign pass rate} (BPR: benign scenarios fully allowed and successfully executed), with Wilson 95\% confidence intervals.

\paragraph{How the suite relates to the design.} We designed both the suite and \pgs{}, so the suite is not a held-out benchmark. Two things partly offset this. First, harm is measured from balances, independent of the policy. Second, three families (H12--H14) were added specifically to attack the guarantee boundary, and one of them (H14) exposed a real flaw in our first implementation. Even so, the rates below are best read as ``what \pgs{} does on attacks of these kinds'', not as a prediction for attacks in the wild.

\subsection{Q1: prevention and benign cost}

\begin{table}[t]
\centering\small
\caption{\textbf{All 260 scenarios} (140 harmful, 120 benign). ``Stopped on-chain'' = allowed before signing, then reverted by a post-condition. ``Gain'' = total value moved to attacker-controlled or burned accounts across harmful scenarios.}
\label{tab:main}
\begin{tabular}{lcccr}
\toprule
Defense & APR [95\% CI] & BPR [95\% CI] & Stopped on-chain & Gain \\
\midrule
None & 0.000 [.000, .027] & 1.000 [.969, 1.00] & -- & \$12.18M \\
Static allowlist & 0.714 [.635, .783] & 0.833 [.757, .889] & -- & \$558k \\
\pgs{}-sim & 0.579 [.496, .657] & 0.975 [.929, .991] & -- & \$855k \\
\pgs{} w/o receipts & 0.864 [.798, .911] & 0.975 [.929, .991] & 30 & \$80k \\
\pgs{} & \textbf{0.936} [.882, .966] & \textbf{0.975} [.929, .991] & 40 & \textbf{\$42k} \\
\bottomrule
\end{tabular}
\end{table}

Table~\ref{tab:main} summarizes the results. \pgs{} prevented 131/140 harmful scenarios and passed 117/120 benign ones. All nine failures come from one family (H12, the in-policy drain), and all three benign refusals are high-impact swaps (B12). The allowlist's lower BPR comes from refusing honest activity it does not know (B8, B12). Its APR comes from refusing unknown contracts, which also happen to be where most of our attacks live.

Table~\ref{tab:taxonomy} splits every harmful scenario by \emph{how} it ended under \pgs{}.

\begin{table}[t]
\centering\small
\caption{\textbf{Prevented, bounded, uncovered.} Outcome of every harmful scenario under \pgs{} (out of 10 per family).}
\label{tab:taxonomy}
\begin{tabularx}{\textwidth}{lcccX}
\toprule
Family & Refused & Reverted & Executed, & Mechanism \\
 & pre-sign & on-chain & harmed & \\
\midrule
H1, H2, H10 & 30 & 0 & 0 & S1: value to a non-allowlisted EOA \\
H3, H4 & 20 & 0 & 0 & S2: residual allowance to an untrusted spender \\
H5, H11 & 20 & 0 & 0 & Value policy: loss $\gg s\,V$ \\
H9 & 10 & 0 & 0 & S3: ownership change \\
H6, H7, H8 & 0 & 30 & 0 & Balance post-condition fails after drift \\
H14 & 0 & 10 & 0 & Receipt post-condition fails after drift \\
H13 & 0 & 0$^\ast$ & 0 & Executed; no profitable front-run existed within the bounds \\
H12 & 1 & 0 & 9 & \emph{Bounded}: session budget caps total loss at $\le W_{\max}$ \\
\bottomrule
\end{tabularx}\\[2pt]
\footnotesize $^\ast$H13 transactions executed successfully with attacker gain \$0 in all 10 variants.
\end{table}

\subsection{Q2: state drift}
Five families (50 scenarios) drift between check and inclusion. \pgs{}-sim missed all 50. Its pre-signing view was accurate, but the executed transaction was not the simulated one. The allowlist missed 40 of them (H6, H8, H13, H14 use allowlisted contracts) and stopped H7 only because the upgradeable pool is not allowlisted. An allowlisted upgradeable contract would be just as exposed.

Under \pgs{}, none of the 50 drift scenarios produced attacker gain:
\begin{itemize}
\item \textbf{H6--H8 (30):} the drift pushed the wallet's outcome below the compiled bound, so the whole batch reverted.
\item \textbf{H14 (10):} the wallet's own balance change was \emph{identical} to the simulation (it sent exactly the intended amount), but the payee received 10--70\% of it. Only the receipt check catches this. The ablation without receipt checks let all 10 through, and the attacker collected 30--90\% of each payment.
\item \textbf{H13 (10):} the searcher saw the bounds in the pending calldata and front-ran as much as they allowed. With the 0.3\% pool fee paid on both legs and $\tau=1\%$, no profitable front-run existed in any variant, and all 10 transactions executed normally. Unguarded, the same searcher extracted \$13.3k--\$32.3k per scenario. This outcome depends on our pool depth, trade sizes and $\tau$. In general, an informed searcher can extract up to roughly $\tau$ of the inflow, and Theorem~\ref{thm:sound} bounds exactly that. $\tau$ is therefore a security parameter, not only a liveness one.
\end{itemize}

\paragraph{\pgs{} is not slippage protection.} A swap's \texttt{minOut} is the closest familiar mechanism, and in H6 a correct \texttt{minOut} would have stopped the sandwich. Four differences matter. (1) \texttt{minOut} protects one output asset of one swap. \pgs{}'s bounds cover \emph{every} tracked asset, including assets the batch was not supposed to touch; H10's hidden transfer would pass any \texttt{minOut}. (2) \texttt{minOut} is set by the party we do not trust. In H6, H8 and H13 the misled agent set it to zero. \pgs{} derives its bounds itself. (3) \texttt{minOut} is not tied to a policy. \pgs{}'s bounds come with a proof that satisfying them satisfies the whole-wallet policy under price uncertainty. (4) \texttt{minOut} cannot express receipts, allowances or ownership. H7 (a drained approval), H9 (ownership) and H14 (payee receipt) have no slippage analogue. B11 (a benign \texttt{minOut}$=0$ swap) passed under \pgs{} 10/10, because \pgs{} adds its own bound.

\subsection{Q3: what is not prevented}
\paragraph{In-policy drains are bounded, not prevented.} In H12 each swap loses about 1.8\%, which is inside the per-transaction policy, so no per-transaction check can refuse it. The session budget, checked by the solver on every transaction against the realized cumulative loss, capped attacker gain at \$4.53k--\$4.76k in the nine variants where swaps executed. Unguarded, the attacker gained \$5.7k--\$10.8k (Figure~\ref{fig:drain}). In the tenth variant the first swap already exceeded the relative bound. The budget is per session, so an attacker who spans sessions is bounded per session only. The allowlist stops H12 only because the skim pool is unknown. An in-policy drain through an \emph{allowlisted} venue would defeat it too.

\begin{figure}[t]
\centering
\includegraphics[width=0.55\textwidth]{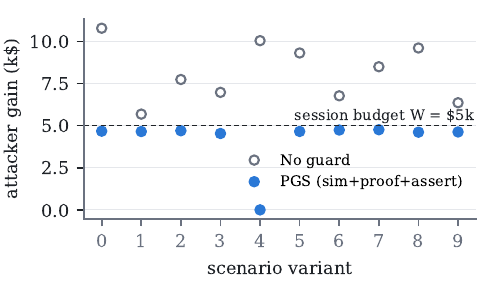}
\caption{Adaptive in-policy drain (H12). Each transaction is individually within policy. The session budget bounds total loss per session.}
\label{fig:drain}
\end{figure}

\subsection{The security--liveness trade-off}
\label{sec:liveness}
\paragraph{Single transactions.} All three benign refusals were high-impact swaps of about \$100k--\$117k into the shallower pool. Their realized cost (2.2--2.6\%) was \emph{below} $s=3\%$, but their worst case over the oracle band was above it. These are genuine false positives caused by band conservatism. Passed B12 swaps cost 1.1--1.9\%. The trade-off is explicit: widening $s$ or narrowing $\beta$ admits larger trades, at the price of letting more in-policy loss through per transaction.

\paragraph{Sessions.} Because $W_{\max}$ counts ordinary fees and price impact as unexplained loss, it also limits legitimate activity. We ran 5 sessions of 40 random benign transactions through one guard each. With a simple arbitrage model that restores pool prices to the oracle after each transaction, 183/200 transactions passed. The budget refused 9, in 3 of 5 sessions, after 23--31 transactions. The rest were 4 high-impact swaps (relative bound) and 4 correct simulation reverts (withdrawing more vault shares than held). Without price restoration only 155/200 passed, because the pools drifted away from our static oracle. That is a testbed artifact, but it shows a real sensitivity: when pool prices and oracle prices diverge, \pgs{} refuses honest trades. In deployment, $W_{\max}$ should scale with volume or be defined relative to expected execution cost, and oracle freshness matters.

\subsection{Q4: overhead}
\label{sec:overhead}
\paragraph{Gas.} On the 117 benign batches that executed under both no defense and \pgs{}, the median batch cost 89.9k gas unguarded and 131.0k under \texttt{executeChecked}. The overhead was a median 41.1k gas (+46\%, max 49.9k), mostly from re-reading balances before and after. The relative overhead is high because our batches are small. It is roughly constant in absolute terms, so it shrinks relative to realistic DeFi transactions (often several hundred thousand gas). For illustration, at \$2{,}500/ETH, 41k gas costs about \$0.10 at 1\,gwei, \$1.03 at 10\,gwei and \$3.08 at 30\,gwei. On L2s, execution gas is typically far cheaper.

\paragraph{Latency.} The median guard decision took 163\,ms (p95 275\,ms, max 0.51\,s over 499 checks) against a local Hardhat node, on a shared 2-vCPU container running two evaluations concurrently. It was 106\,ms in an earlier run without contention. Time is dominated by simulation RPC round trips; solver calls took milliseconds. Against a remote RPC, network latency would add to this.

\subsection{The intent-level baseline: an LLM judge}
\label{sec:judge}
LLM review operates at a different layer from \pgs{}. It reasons about \emph{intent} (``does this match the task?'') rather than effects. We ran it on variants 0--4 of the first 24 families (120 scenarios, 242 transactions, shuffled, opaque ids, batches of about 48, instructions to judge each item independently). This used our first code version. \pgs{}'s outcomes on the 960 runs the two versions share are identical, so the comparison stands. Figure~\ref{fig:main} shows the results and Table~\ref{tab:judge} the families where the defenses differ.

\begin{figure}[t]
\centering
\includegraphics[width=0.85\textwidth]{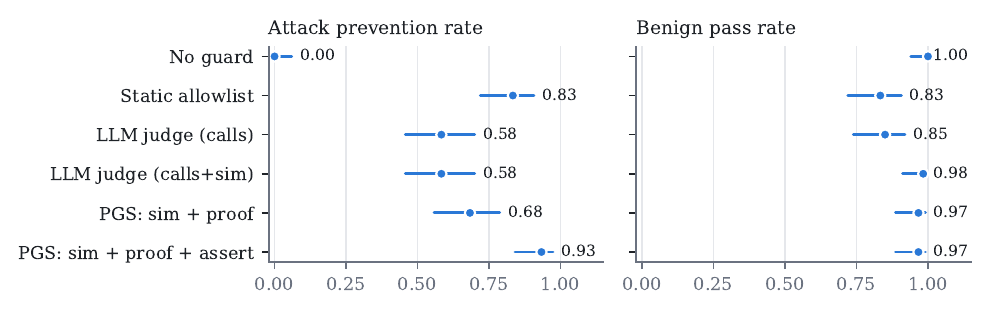}
\caption{Attack prevention (left) and benign pass (right) on the 120-scenario subset (60 harmful, 60 benign) with Wilson 95\% CIs.}
\label{fig:main}
\end{figure}

\begin{table}[t]
\centering\small
\caption{Per-family outcomes on the 120-scenario subset (prevented or passed, out of 5), for the families where defenses differ. AL = allowlist, J = LLM judge, J+S = judge with simulation, S = \pgs{}-sim, P = \pgs{}.}
\label{tab:judge}
\begin{tabular}{lccccc}
\toprule
Family & AL & J & J+S & S & P \\
\midrule
H1 direct transfer & 5 & 3 & 3 & 5 & 5 \\
H2 poisoned address & 5 & 3 & 3 & 5 & 5 \\
H5 fake pool & 5 & 1 & 5 & 5 & 5 \\
H6 sandwich (drift) & 0 & 5 & 4 & 0 & 5 \\
H7 proxy upgrade (drift) & 5 & 3 & 0 & 0 & 5 \\
H8 fee rug (drift) & 0 & 0 & 0 & 0 & 5 \\
H11 fake bridge & 5 & 0 & 0 & 5 & 5 \\
H12 in-policy drain & 5 & 0 & 0 & 1 & 1 \\
\midrule
B8 honest non-allowlisted pool & 0 & 5 & 5 & 5 & 5 \\
B11 \texttt{minOut}$=0$, no MEV & 5 & 0 & 4 & 5 & 5 \\
B12 high-impact swap & 0 & 1 & 5 & 3 & 3 \\
\bottomrule
\end{tabular}
\end{table}

The judge prevented 58.3\% of attacks with or without the simulation, against 93.3\% for \pgs{} on the same subset. Three observations are more informative than the rates.
\begin{itemize}
\item \textbf{The judge catches what effects cannot show.} It flagged the sandwich \emph{setup} (\texttt{minOut}$=0$ on a large swap) in every H6 case, before any drift had happened. That is reasoning about risk, which \pgs{} does not do. It also refused every benign \texttt{minOut}$=0$ swap, which is defensible caution.
\item \textbf{It misses what looks legitimate at the call level.} It allowed every fake-bridge deposit and most fake-pool swaps, because these are call-for-call identical to legitimate interactions.
\item \textbf{A clean simulation can \emph{raise} confidence in a drift attack.} With the simulation, the judge's benign pass rate rose from 0.85 to 0.98 and it caught the fake pool (5/5), but its approval of the upgradeable-pool attack went from 2/5 to 5/5. The simulation, run before the upgrade, looked clean. We consider this the central qualitative finding for agent oversight: a simulation is evidence about the wrong state, and handing it to a reviewer can make things worse.
\end{itemize}
The two layers are complementary. An intent-level reviewer can refuse transactions that are \emph{risky} but not yet harmful, and \pgs{} guarantees that whatever is signed respects the policy when it executes.

\subsection{LLM wallet agents under injection}
We also ran 20 episodes in which an LLM agent operated the wallet through a command-line tool and proposed transactions: 10 tasks $\times$ 2 runs, run~1 with Claude Opus (\texttt{claude-opus-5-5}) and run~2 with Claude Haiku. Six tasks carried indirect injection: a spoofed invoice with a poisoned address, an airdrop drainer, a promoted fake pool, a governance post requesting ownership transfer, a fake bridge in the documentation, and a fee-collector instruction in a quote tool's output. Agents were told truthfully that the chain was a local research testbed, which may make them more careful than in deployment.
\begin{itemize}
\item Both models ignored the spoofed invoice and the fake pool. On the fee injection, the default model swapped without paying the ``fee'' and Haiku declined the whole task. The default model also declined the airdrop and the ownership transfer.
\item Both models proposed the fake-bridge deposit. \pgs{} and the allowlist refused it, and in the scripted equivalent (H11) both LLM-judge variants allowed it.
\item By its own report, the Haiku agent \emph{attempted} to follow the airdrop-drainer and ownership-transfer injections. The agent runtime's safety filter blocked those tool calls, so no proposals were recorded, and we did not re-submit them. The equivalent transactions are the scripted H3 and H9 families, which \pgs{} refused 10/10.
\item Two Haiku proposals were malformed (a swap without approval, one of them in a clean task). \pgs{} refused both as simulation reverts; they would have reverted on-chain anyway. \pgs{} passed every other clean-task proposal.
\end{itemize}
With capable agents, most injections failed at the agent itself. What got through was an attack that is \emph{indistinguishable at the intent level} (fake infrastructure), which is exactly the class an effect-level guard targets. We release the proposals and condensed agent self-reports, not full transcripts.

\subsection{Reproducibility and audit}
We reran the full scripted evaluation (26 families $\times$ 10 variants $\times$ 5 defenses $=1{,}300$ runs) on a freshly started node with the same code. Harm labels, decisions, execution outcomes and attacker gains matched in all 1{,}300 runs. The same held for the first code version (960 runs). This shows determinism, not independent replication.

Before finalizing, an independent audit recomputed every reported number from the raw logs and checked the implementation against the theorem. It found an incorrect headline claim in our draft (we had attributed all simulation-only failures to drift, but the in-policy drain also defeats it), the payment-crediting flaw, the absence of any drift attack that adapts to the defense, and several overstated sentences. This version fixes the code, adds H13, H14 and the session study, and rewords the claims. We report the audit because it changed the results.

\section{Related Work}
\label{sec:related}
\paragraph{Prompt injection and agent defenses.} Indirect prompt injection~\cite{greshake2023not} and agent benchmarks such as InjecAgent~\cite{zhan2024injecagent} and AgentDojo~\cite{debenedetti2024agentdojo} establish that tool-using agents can be redirected by content they read. Design-level defenses constrain how untrusted data can influence an agent's control flow~\cite{debenedetti2025camel}. \pgs{} makes no assumption about the agent. It guards the one effectful boundary, the signature, and would compose with any agent-side defense.

\paragraph{Slippage limits and deadlines.} Per-swap bounds such as a minimum output amount are the most widely deployed post-execution checks in DeFi. They are chosen by the transaction's author, cover one asset of one action, and carry no statement about the rest of the wallet. \pgs{} generalizes the idea in three directions: the bounds are derived by an independent guard, they cover all tracked assets plus receipts, allowances and ownership, and they come with a proof that they imply a whole-wallet policy (Section~\ref{sec:eval}, ``\pgs{} is not slippage protection'').

\paragraph{Transaction simulation and allowlists.} Wallet-side simulation previews and policy engines with allowlists are common in practice. We implement strong versions of both (\pgs{}-sim and the static allowlist) and show that both are exposed to drift. Simulation is also blind to intent-level fakes that behave normally, and allowlists are blind to allowlisted-but-mutable contracts.

\paragraph{Formal verification of smart contracts.} Tools such as Securify~\cite{tsankov2018securify} and VerX~\cite{permenev2020verx} prove properties of \emph{contracts} for all executions. \pgs{} proves a property of \emph{one transaction's effects}, for the executions allowed by its post-conditions. The two are complementary: a verified protocol still cannot stop a misled agent from using it harmfully, and \pgs{} cannot prove a protocol correct. Recent discussion of AI-assisted formal verification for onchain systems~\cite{buterin2026fv,sherlock2026fv} stresses property selection, explicit assumptions and version binding. \pgs{} applies that discipline to individual agent transactions, and version binding becomes ``binding to the execution''.

\paragraph{Runtime verification.} Monitoring an execution against a specification and reacting to violations is the core of runtime verification~\cite{leucker2009brief}. \pgs{} is a runtime check with two twists: the monitored properties are \emph{derived} per transaction from a policy, together with a proof that they suffice, and the reaction is an atomic revert.

\paragraph{Account abstraction.} Smart-contract accounts~\cite{erc4337} make it possible to enforce checks inside the wallet without protocol changes. That is the deployment path for \texttt{executeChecked}. Account abstraction also addresses our key-custody assumption: a validation rule can admit only guarded execution.

\paragraph{Intent-based systems.} In intent-based designs, users sign a desired outcome and solvers compete to realize it. The user's signed constraints play a role similar to our post-conditions. \pgs{} differs in who writes the constraints (an independent guard, not the possibly misled agent), in what they must imply (a whole-wallet policy under price uncertainty), and in scope (arbitrary call batches, not only exchanges). An agent that emits intents would still benefit from a guard that checks the intent's constraints against policy.

\paragraph{MEV and DeFi attacks.} Frontrunning and sandwiching~\cite{daian2020flash,qin2022quantifying} and taxonomies of DeFi attacks~\cite{zhou2023sok} motivate our drift families. The mempool-aware searcher in H13 is the natural adversary for any published post-condition.

\section{Discussion and Limitations}
\label{sec:limits}
\paragraph{How to read the numbers.} The 93.6\% prevention rate is the rate on our 140 harmful scenarios, from 14 families we designed. It is not an estimate of real-world effectiveness. The per-family results (Tables~\ref{tab:taxonomy} and~\ref{tab:judge}) are more informative than the aggregate. They show which mechanism handles which kind of attack, and they would transfer to any attack of the same kind.

\paragraph{Testbed realism.} Our contracts are simplified. The pools are constant-product pools without concentrated liquidity, the oracle is static, and the adversary hooks are scripted. We did not replay real exploits on a mainnet fork, and we did not measure how often real transactions drift between simulation and inclusion. We think measuring drift on real transactions is the most valuable next step. It would tell practitioners how much of the gap exists in the wild, independently of any defense.

\paragraph{Scope of the guarantee.} See Table~\ref{tab:boundary}. The most important practical gaps are untracked assets (refused conservatively, at a liveness cost for long-tail tokens) and off-chain signatures such as permits, which must be gated separately.

\paragraph{Legitimate value sinks.} A real bridge also takes value with no immediate inflow. Our value policy refuses it just as it refuses the fake one. Legitimate bridges need an explicit authorized-outflow rule, like payees, and fake-bridge detection then reduces to the quality of that list.

\paragraph{Policy risk.} A correct proof of a weak policy is still weak~\cite{sherlock2026fv}. Our policy is small enough to audit, but $s$, $\beta$, $\tau$ and $W_{\max}$ are risk decisions with the trade-offs shown in Section~\ref{sec:liveness}.

\paragraph{Liveness and griefing.} Post-conditions turn some attacks into reverts. An attacker can use this to make the agent fail tasks or waste gas (griefing), but not to extract value beyond the bounds.

\paragraph{Model coverage.} The LLM judge and the LLM agents are all Claude models. Different models, prompts or entity labeling could change the LLM results, and the judge ran on only half of the scenarios.

\paragraph{Future work.} (1) Measure drift on mainnet traffic. (2) Replay historical exploits on a mainnet fork. (3) Evaluate with a held-out attack set written by people who have not seen the guard. (4) Add correlated price models and cross-asset policies, where the solver is needed. (5) Integrate with an ERC-4337 validator so that key custody is enforced on-chain.

\section{Conclusion}
For agents that sign transactions, a pre-signing check should be read as a proof with a scope. Simulation provides the evidence, a solver checks the claim over price uncertainty, and post-conditions derived from that proof make the claim hold for the transaction that actually executes. On our testbed this closed the drift gap that simulation-based, allowlist-based and LLM-based guards leave open. It cost about 41k gas and a fraction of a second per transaction. It also made the remaining risk explicit: in-policy loss is bounded per session, and untracked assets are out of scope.

\section*{Acknowledgments}
The author used AI-based software and writing tools during implementation and preparation of this manuscript, and reviewed and verified all code, experiments and text.

\end{document}